\documentclass[conference]{IEEEtran}
\usepackage{enumerate}
\usepackage{graphicx}
\usepackage{epsf}
\usepackage{subfigure}
\usepackage{amsmath}
\usepackage{amssymb}
\usepackage{array}
\usepackage{setspace}
\usepackage[amsmath,thmmarks,amsthm]{ntheorem}
\usepackage[ruled,lined]{algorithm2e}
\usepackage{algorithmic}
\usepackage{color}
\usepackage{amsfonts}
\usepackage{dsfont}
\usepackage{xfrac}
\usepackage{breqn}
\usepackage{bbm}
\usepackage{cite}

\graphicspath{{fig/}}

\theoremseparator{.}

\newtheorem{theorem}{Theorem}

\begin{document}

\title{	Reconfigurable Intelligent Surface-Empowered Self-Interference Cancellation for 6G Full-Duplex MIMO Communication Systems }

\author{\IEEEauthorblockN{Chia-Jou Ku, Li-Hsiang Shen and Kai-Ten Feng}
\IEEEauthorblockA{
Department of Electronics and  Electrical Engineering\\ 
National Yang Ming Chiao Tung University, Hsinchu, Taiwan\\ 
chiajouku.ee09@nycu.edu.tw, 
gp3xu4vu6.cm04g@nctu.edu.tw
and ktfeng@nycu.edu.tw}}
\maketitle

\begin{abstract}
	Substantially increasing wireless traffic and extending serving coverage is required with the advent of sixth-generation (6G) wireless communication networks. Reconfigurable intelligent surface (RIS) is widely considered as a promising technique which is capable of improving the system sum rate and energy efficiency. Moreover, full-duplex (FD) multi-input-multi-output (MIMO) transmission provides simultaneous transmit and received signals, which theoretically provides twice of spectrum efficiency. However, the self-interference (SI) in FD system is a challenging task requiring high-overhead cancellation, which can be resolved by configuring appropriate phase shifts of RIS. This paper has proposed an RIS-empowered full-duplex interference cancellation (RFIC) scheme in order to alleviate the severe interference in an RIS-FD system. We consider the interference minimization of RIS-FD MIMO while guaranteeing quality-of-service (QoS) of whole system. The closed-form solution of RIS phase shifts is theoretically derived with the discussion of different numbers of RIS elements and receiving antennas. Simulation results reveal that the proposed RFIC scheme outperforms existing benchmarks with more than 50$\%$ of performance gain of sum rate.

\end{abstract}

\section{Introduction}

	 The requirement of wireless data demands is increasingly high as the sixth-generation (6G) technology evolves. To meet the compellingly high traffic demands, there are tremendous emerging techniques for 6G wireless communications. Reconfigurable intelligent surface (RIS) is widely considered as one of the promising 6G techniques for extending coverage area, reducing power consumption, and enhancing system performance. The RIS is composed of numerous ultra-thin metamaterial-based elements, which can reflect the received signals on the surface without additional signal processing. Moreover, as a benefit of cost-effective RIS, it can be readily employed by reconfiguring its phase shifts to alternate the original channel between transmitters and receivers in order to increase the service coverage and spectrum-energy efficiency \cite{2,3,6,14}. 
	 
	 RIS-assisted 6G wireless communications have been widely discussed due to its cost-effective and unsophisticated deployment which can be utilized in various systems, e.g., signal quality is improved by RIS under non-orthogonal multiple access system \cite{10} and multi-input-multi-output (MIMO) networks \cite{11}. Another arising critical problem of RIS is the joint design of precoding at transmitters and phase shifts of RIS elements \cite{7,8,15}. In \cite{7}, the authors aim for minimizing the symbol error rate by deploying RIS, whilst authors in \cite{8} maximize the energy efficiency of the reflected signals. By adjusting RIS configuration, the authors in \cite{15} tend to minimize total base station (BS) power consumption. In addition, enhanced security can be reached by converting the original intruded wireless paths with RIS deployment \cite{9,12}. In \cite{100,101}, theoretical analysis of service coverage probability have been conducted with RIS deployment. Moreover, RIS can be utilized for interference management, which is able to suppress the strong interference and to enhance the desired signal strength by adjusting RIS phase shifts in different transmission directions \cite{16}.
	 
	 In order to further increase the spectrum efficiency, full-duplex (FD) communication can simultaneously transmit and receive signals at a single operating frequency, which theoretically provides twice of spectrum efficiency compared to the half-duplex technique. However, the main challenging issue is that FD possesses strong self-interference (SI) degrading the signal quality, which results in lowered sum rate. Therefore, existing literatures have researched on SI mitigation via signal processing methods \cite{26,4,1,5,13}. In \cite{4}, they minimize severe SI power by deriving the optimal eigen-based beamforming through spatial-domain elimination. SI mitigation with oblique projection is applied in an FD system \cite{1}. To maximize the vehicular-to-infrastructure networks under the FD system, a resource allocation method is proposed in \cite{5}. Moreover, in \cite{13}, a resource management problem based FD networks with consideration of user equipment's (UE's) queue backlog and time-varying channels is discussed. However, the above-mentioned works require additional overhead of channel estimation and signal processing for FD SI cancellation. Therefore, it becomes compellingly imperative to design efficient RIS to alternate wireless channels in order to alleviate interference in an RIS-FD system. 
	 
	 In this paper, we conceive an FD MIMO transmission which is empowered by RIS and aim for maximizing total system sum rate in an RIS-FD network while guaranteeing both downlink/uplink (DL/UL) quality-of-service (QoS) and RIS phase constraint.
	 The contributions of this work are summarized as follows.
\begin{itemize}
	
	\item We have proposed to formulate and derive the solution of RIS phase shifts for maximizing system sum rate with the guarantee of DL/UL signal quality in the RIS-FD MIMO systems, considering both DL SI and UE's UL co-channel interference.
	\item The proposed RIS-empowered full-duplex interference cancellation (RFIC) scheme can entirely remove the DL SI and UL UE's co-channel interference when the number of RIS elements is equal to that of BS transmit antennas and DL UEs. The corresponding optimal phase shift of RIS is theoretically proved in a closed-form manner.
	
	\item Simulations have revealed that the proposed RFIC scheme can largely enhance the sum rate performance by cancelling SI and co-channel interference under proposed RIS-FD system in terms of different configurations. The comparison also demonstrates that our RFIC scheme outperforms the other existing interference cancellation methods in open literatures with more than 50$\%$ of performance gain of sum rate.
\end{itemize} 
	
\addtolength{\topmargin}{0.055 in}
\section{System Model And Problem Formulation} \label{SM}

	We consider an RIS-assisted FD MIMO system as illustrated in Fig. \ref{Fig.7}. An RIS-FD MIMO BS consists of $N_t$ transmit antennas and $N_r$ receiving antennas. Note that we deploy a single RIS to empower the data transmission which is equipped with $K$ reflecting elements. There are $N$ UEs intending UL transmission and $M$ UEs desired for DL reception operating at the identical frequency bands, where each UE is assumed to be equipped with a 	single antenna. DL and UL transmission paths are categorized by direct path and reflected paths. For UL, ${\rm\mathbf{U} \in \mathbb{C}}^{N_r\times{N}}$ is the direct path from UL UEs to BS receiving antennas, while the reflected paths from UL UEs to RIS and from RIS to BS are respectively defined as ${\rm\mathbf{U_1} \in \mathbb{C}}^{K\times{N}}$ and ${\rm\mathbf{U_2} \in \mathbb{C}}^{N_r\times{K}}$. For DL, ${\rm\mathbf{D} \in \mathbb{C}}^{M\times{N_t}}$ is the direct path from BS transmit antennas to DL UEs, and the reflected paths from BS transmit antennas to RIS and from RIS to DL UEs are respectively defined as ${\rm\mathbf{D_1} \in \mathbb{C}}^{K\times{N_t}}$ and ${\rm\mathbf{D_2} \in \mathbb{C}}^{M\times{K}}$. For this system, the self-interference channel induced by FD MIMO from BS transmit to received antennas is denoted as ${\rm\mathbf{S} \in \mathbb{C}}^{N_r\times{N_t}}$ and the co-channel interference from UL UEs to DL UEs is denoted as ${\rm\mathbf{V} \in \mathbb{C}}^{M\times{N}}$.
\addtolength{\topmargin}{0.105 in}
\begin{figure}
\centering
\includegraphics[width=3.3in]{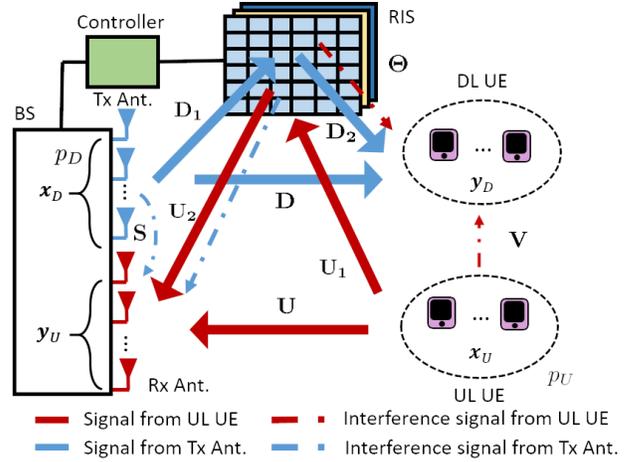}
\caption{RIS-assisted FD MIMO system for 6G wireless communications.} \label{Fig.7}
\end{figure}
As for RIS, it configures the phase shifters for converting transmission directions, which requires zero power without circuit power amplifier. We define 
\begin{equation} \label{RIS}
	\mathbf{\Theta} = \operatorname{diag}(\alpha {e^{j\theta_{1}}},...,\alpha e^{j\theta_{i}},...,\alpha {e^{j\theta_{K}}})
\end{equation}
as the phase shift diagonal matrix of RIS, where $\operatorname{diag}(\cdot)$ is the diagonal operation of a matrix, $\alpha$ is a constant related to reflection decaying efficiency, and $\theta_{i} \in [0,2\pi]$ is the phase of the $i$-th element of RIS. The received DL signal at DL UE and received UL signal at BS receiving antennas can be expressed
\begin{equation} \label{received DL}
	{\rm\mathbf{y}}_D \!=\! \sqrt{p_D}{\rm\mathbf{(D \!+\! D_2 \Theta D_1) x}}_D +\sqrt{p_U}{\rm\mathbf{(V +D_2 \Theta U_1)x}}_U+{\rm\mathbf{n}}_1,
\end{equation}
\begin{equation} \label{received UL}
	{\rm\mathbf{y}}_U \!=\! \sqrt{p_U}{\rm\mathbf{(U \!+\! U_2 \Theta U_1) x}}_U +\sqrt{p_D}{\rm\mathbf{(S +U_2 \Theta D_1)x}}_D+{\rm\mathbf{n}}_2,
\end{equation}
where $p_U$ and $p_D$ are respectively the UL and DL transmit power of UEs and BS. ${\rm\mathbf{x}}_U=[x_{U,1}, x_{U,2},\cdots, x_{U,N}]^T\in \mathbb{C}^{N}$ and ${\rm\mathbf{x}}_D=[x_{D,1}, x_{D,2},\cdots, x_{D,N_t}]^T\in \mathbb{C}^{N_t}$ denote the signals from UL UEs and from DL transmit antennas at BS, respectively, and ${[\cdot]}^{T}$ is the transpose operation. The notations of ${\rm\mathbf{n}_1,\rm\mathbf{n}_2\sim\mathcal{C}\mathcal{N}}({\rm{\mathbf{0}}},N_0 B{\rm{\mathbf{e}}}_{N_r})$ are defined as the complex white Gaussian noises, where $\mathbf{0}$ is the zero-mean vector, ${\rm{\mathbf{e}}}_{N_r}$ denotes the unit vector, $N_0$ is noise power spectral density and $B$ is the system bandwidth. Note that the signal channel includes both the received signal from direct link and reflected signals through the RIS, so does the interference channel.

	By deploying the RIS, we can potentially enhance the quality of signal by dynamically adjusting the RIS phase shifter $\rm\mathbf{\Theta}$ while guaranteeing the performance of DL/UL system and phase constraint of $[0,2\pi ]$. The considered problem for sum rate maximization on both UL and DL can be formulated as
\begin{subequations}\label{obj}
  \begin{align} 
   \quad &\underset{\rm\mathbf{\Theta}}{\text{max}}\quad 
R_U+R_D\\
    &\text{s.t.} \quad \theta_{i} \in [0,2\pi],  \quad \forall \, 1\leq i \leq K, \label{pp1}\\
    &\quad\quad R_U  \ge \gamma_{thr,U},\label{D2}\\
    &\quad\quad R_D \ge \gamma_{thr,D},\label{D3}   
  \end{align}
\end{subequations}
where 
 \begin{equation}
 \begin{small}
  \begin{aligned} 
   & R_U=
\log_2\left(1+\frac{\parallel \sqrt{p_U}{\rm\mathbf{\left(U + U_2 \Theta U_1\right)}} {\rm\mathbf{x}}_U \parallel^2}{\parallel \sqrt{p_D}{\rm\mathbf{\left(S +U_2 \Theta D_1\right)}}{\rm\mathbf{x}}_D \parallel^2 +\parallel{\rm\mathbf{n}}_1\parallel^2}\right),   
  \end{aligned}
  \end{small}
\end{equation}	
 \begin{equation}
  \begin{small}
  \begin{aligned} 
& R_D= \log_2\left(1+\frac{\parallel \sqrt{p_D}{\rm\mathbf{\left(D + D_2 \Theta D_1\right)}} {\rm\mathbf{x}}_D \parallel^2}{\parallel\sqrt{p_U}{\rm\mathbf{\left(V +D_2 \Theta U_1\right)x}}_U\parallel^2 +\parallel{\rm\mathbf{n}}_2\parallel^2}\right).    
  \end{aligned}
  \end{small}
\end{equation}	
The notations $R_U\in\mathbb{R}$ and $R_D\in\mathbb{R}$ represents the UL and DL sum rate, respectively. The constraint $\eqref{D2}$ represents that the UL sum rate $R_U$ is required to be larger than a pre-defined threshold $\gamma_{thr,U}$ for guaranteeing UL QoS, whilst $\eqref{D3}$ indicates that for DL QoS satisfaction.
 The proposed optimization problem is transformed as
\begin{subequations}\label{obj1}
  \begin{align} 
   &\underset{{\rm\mathbf{\Theta}},\mu}{\text{min}} \quad \parallel \sqrt{p_D}{\rm\mathbf{(S +U_2 \Theta D_1)}}{\rm\mathbf{x}}_D \parallel^2\notag \\&\qquad\quad +\mu\parallel\sqrt{p_U}{\rm\mathbf{(V +D_2 \Theta U_1)x}}_U\parallel^2 \label{D1}\\
    &\text{s.t.} \quad \theta_{i} \in [0,2\pi],  \quad \forall \, 1\leq i \leq K, \\
    &\quad\quad \parallel \sqrt{p_U}{\rm\mathbf{(U + U_2 \Theta U_1)}} {\rm\mathbf{x}}_U \parallel^2 \ge t_{thr,U},\label{C2}\\
    &\quad\quad \parallel \sqrt{p_D}{\rm\mathbf{(D + D_2 \Theta D_1)}} {\rm\mathbf{x}}_D \parallel^2 \ge t_{thr,D},\label{C3}\\
    &\quad\quad \mu \in \mathbb{R}.
  \end{align}
\end{subequations}	 
However, the problem objective and constraints of $\eqref{obj}$ are non-convex and non-linear. Note that the sum rates $R_U$ and $R_D$ increase monotonically with the signal-to-noise-ratio (SINR) term because $\log_2$ is a monotone increasing function, i.e., sum rate $={\log_2}(1+$SINR$)$. While the received signal power is guaranteed to be above a certain value, maximizing the SINR of received signal is approximately equivalent to minimizing the interference power. The objective function can be transformed into minimizing the power of UL SI and DL co-channel interference and guaranteeing the power of DL/UL received signal.

The parameter $\mu$ represents the importance ratio between DL co-channel interference and UL SI. Note that $t_{thr,U}$ and $t_{thr,D}$ in $\eqref{C2}$ and $\eqref{C3}$ respectively indicate UL/DL signal quality satisfaction, which implies asymptotic meaning of QoS satisfaction with those in $\eqref{D2}$ and $\eqref{D3}$. In the following, we will discuss two different cases with and without QoS consideration, i.e., $t_{thr,U} \neq 0$ and  $t_{thr,D} \neq 0$ for QoS-aware scheme and $t_{thr,U} = t_{thr,D} = 0$ for non-QoS scenario design.
\addtolength{\topmargin}{+0.21 in}
\section{Proposed RIS-Empowered FD Interference Cancellation (RFIC) Scheme} \label{BP}
	In the conventional FD MIMO system without RIS deployment, the existing works utilize signal processing methods via either precoding or postcoding to suppress the interference, which requires unaffordable computational overhead in practical implementation with increasing number of transmit or receiving antennas. However, under RIS deployment, we can generate the artificial channel to automatically and adjustably mitigate strong interferences. Our proposed RFIC scheme can optimally determine the RIS phase to substantially mitigate FD interferences. 
	
	As observed from $\eqref{obj1}$, we can infer that both the objective and corresponding constraints possess convexity property indicating that the optimal solution can be derived via celebrated Lagrange optimization \cite{opti}. By applying Lagrangian method, we can obtain the augmented Lagrangian expression with parameters of ${\rm\boldsymbol\theta}_{t}$ and ${\rm\boldsymbol\lambda}_{t}$ as
	\begin{equation}\label{Lag3}
	\begin{aligned}
J({\rm\boldsymbol\theta}_t,{\rm\boldsymbol\lambda}_t)&= P({\rm\boldsymbol\theta}_t,{\rm\boldsymbol\lambda}'_t)\\&
- \lambda_{2K+1}(\parallel\sqrt{p_U}{\rm\mathbf{(U + U_2 \Theta U_1)}}{\rm\mathbf{x}}_U \parallel^2-t_{thr,U})\\&
- \lambda_{2K+2}(\parallel\sqrt{p_D}{\rm\mathbf{(D + D_2 \Theta D_1)}}{\rm\mathbf{x}}_D \parallel^2-t_{thr,D}),
\end{aligned}
\end{equation}\vspace{-5pt}
where
\begin{equation}
\begin{aligned}
P({\rm\boldsymbol\theta}_t,{\rm\boldsymbol\lambda}'_t)&=\quad 
\parallel\sqrt{p_D}{\rm\mathbf{(S +U_2 \Theta D_1)}}{\rm\mathbf{x}}_D\parallel^2\\&+\mu\parallel\sqrt{p_U}{\rm\mathbf{(V +D_2 \Theta U_1)x}}_U\parallel^2\\ &- \left(\sum_{i=0}^{K-1} \lambda_{2i+1}\theta_i+\lambda_{2i+2}(\theta_i-2\pi)\right).
\end{aligned}
\end{equation}
The parameter ${\boldsymbol{\theta}}_{t}=\{\theta_1, \theta_2, ..., \theta_{K}\}$ is a set containing $K$ designed RIS variables.
 ${\boldsymbol{\lambda}}_{t}=\{\lambda_1, \lambda_2, ..., \lambda_{2K+1},\lambda_{2K+2}\}$ and ${\boldsymbol{\lambda}}'_{t}=\{\lambda_1, \lambda_2, ...,  \lambda_{2K}\}$ are the set of $2K+2$ and $2K$ Lagrangian multipliers. The first-order derivative of the Lagrangian expression of $\eqref{Lag3}$ can be given in $\eqref{Lag4}$ shown at the top of next page,
\begin{figure*}
\vspace{2pt}
\begin{small}
\begin{equation}
\begin{aligned}\label{Lag4}
\frac{\partial J({\rm\boldsymbol\theta}_t,{\rm\boldsymbol\lambda}_t)}{\partial \theta_i}=&
	N_i({\rm\boldsymbol\theta}_t,{\rm\boldsymbol\lambda}'_t)
	-\lambda_{2K+1}\left(p_U x_U^2\right)2jb_{n,i}e^{j\theta_i}\sum_{n=1}^{N_r}  \left( b_{n,i} +\sum_{p\neq i}^{N_r} b_{n,p}e^{j\theta_p}+{\rm\mathbf{u}}_n \right)
	-\lambda_{2K+2}\left(p_D x_D^2\right)2ja_{m,i}e^{j\theta_i}\\&\sum_{m=1}^{M}\left(a_{m,i} +\sum_{q\neq i}^{M} a_{m,q}e^{j\theta_q}+{\rm\mathbf{d}}_n \right)
	,  \quad \forall \, 1\leq i \leq K,
\end{aligned}
\end{equation}
where 
\begin{equation}
\begin{aligned}
	N_i({\rm\boldsymbol\theta}_t,{\rm\boldsymbol\lambda}'_t)=&
	\left(p_D x_D^2\right)2jz_{n,i}e^{j\theta_i} \sum_{n=1}^{N_r}\left(z_{n,i}+\sum_{p\neq i}^{N_r} z_{n,p}e^{j\theta_p}+{\rm\mathbf{s}}_n \right)+\mu \left(p_U x_U^2\right)2jy_{m,i}e^{j\theta_i}\sum_{m=1}^{M}
	\left(y_{m,i}+\sum_{q\neq i}^{M} y_{m,q}e^{j\theta_q}+{\rm\mathbf{v}}_n \right)\\&-\lambda_{2i-1}-\lambda_{2i}
	,  \quad \forall \, 1\leq i \leq K.
\end{aligned}
\end{equation}
\end{small}
\hrulefill
\end{figure*}
where $z_{i,j} = {\rm\mathbf{U}_2}_{i,j}{\rm\mathbf{d}_1}_{j}$,  ${\rm\mathbf{d}_1}\!=\! [\sum_{i=1}^{N_t} {\rm\mathbf{D}}_{1_{1,i}}, \sum_{i=1}^{N_t}{\rm\mathbf{D}}_{1_{2,i}}, ..., \\\sum_{i=1}^{N_t} {\rm\mathbf{D}}_{1_{K,i}}]^T  \!\in\! \mathbb{C}^{K}$ is the sum of the columns of ${\rm\mathbf{D_1}}$, $a_{i,j} = \\{\rm\mathbf{D}_2}_{i,j}{\rm\mathbf{d}_1}_{j}$, where ${\rm\mathbf{U}_2}_{i,j}$, ${\rm\mathbf{D}_1}_{i,j}$ and ${\rm\mathbf{D}_2}_{i,j}$ are the elements in the $i$-th row and the $j$-th column of ${\rm\mathbf{U}_2}$, ${\rm\mathbf{D}_1}$ and ${\rm\mathbf{D}_2}$, and ${\rm\mathbf{d}_1}_{j}$ is the $j$-th element of ${\rm\mathbf{d}_1}$. The notation $b_{i,j} = {\rm\mathbf{U}}_{2_{i,j}}{\rm\mathbf{u}_1}_{j}$ is defined with ${\rm\mathbf{u}_1} = [ \sum_{i=1}^{N}{\rm\mathbf{U}}_{1_{1,i}}, \sum_{i=1}^{N}{\rm\mathbf{U}}_{1_{2,i}}, ..., \sum_{i=1}^{N}{\rm\mathbf{U}}_{1_{K,i}} ]^T\in \mathbb{C}^{K}$ denoted as the sum of the columns of ${\rm\mathbf{U_1}}$, and $y_{i,j} = {\rm\mathbf{D}_2}_{i,j}{\rm\mathbf{u}_1}_{j}$, where ${\rm\mathbf{U}}_{1_{i,j}}$ is the element in the $i$-th row and the $j$-th column of ${\rm\mathbf{U}_1}$, and ${\rm\mathbf{u}_1}_{j}$ is the $j$-th element of ${\rm\mathbf{u}_1}$. The notation ${\rm\mathbf{s}} = [\sum_{i=1}^{N_t} {\rm\mathbf{S}}_{1,i}, \sum_{i=1}^{N_t}{\rm\mathbf{S}}_{2,i}, ..., \sum_{i=1}^{N_t} {\rm\mathbf{S}}_{N_r,i}]^T\in \mathbb{C}^{N_r}$ is the sum of the columns of ${\rm\mathbf{S}}$, ${\rm\mathbf{v}} = [\sum_{i=1}^{N} {\rm\mathbf{V}}_{1,i},  \sum_{i=1}^{N}{\rm\mathbf{V}}_{2,i}, ...,  \sum_{i=1}^{N}\\ {\rm\mathbf{V}}_{M,i}]^T\in \mathbb{C}^{M}$ is the sum of the columns of ${\rm\mathbf{V}}$, where ${\rm\mathbf{S}}_{i,j}$ and ${\rm\mathbf{V}}_{i,j}$ are the elements in the $i$-th row and the $j$-th column of ${\rm\mathbf{S}}$ and ${\rm\mathbf{V}}$, respectively. Similarly, ${\rm\mathbf{u}} = [\sum_{i=1}^{N} {\rm\mathbf{U}}_{1,i},  \sum_{i=1}^{N}{\rm\mathbf{U}}_{2,i}, ...,  \\\sum_{i=1}^{N} {\rm\mathbf{U}}_{N_r,i}]^T\in \mathbb{C}^{N_r}$ is the sum of the columns of ${\rm\mathbf{U}}$ and ${\rm\mathbf{d}} = [\sum_{i=1}^{N_r} {\rm\mathbf{D}}_{1,i}, \sum_{i=1}^{N_t}{\rm\mathbf{D}}_{2,i}, ...,  \sum_{i=1}^{N_t} {\rm\mathbf{D}}_{M,i}]^T\in \mathbb{C}^{M}$ is the sum of the columns of ${\rm\mathbf{D}}$, where ${\rm\mathbf{U}}_{i,j}$ and ${\rm\mathbf{D}}_{i,j}$ are the elements in the $i$-th row and the $j$-th column of ${\rm\mathbf{U}}$ and ${\rm\mathbf{D}}$, respectively. ${\rm\mathbf{s}}_n$, ${\rm\mathbf{v}}_n$, ${\rm\mathbf{u}}_n$ and ${\rm\mathbf{d}}_n$ respectively represent the $n$-th element of ${\rm\mathbf{s}}$, ${\rm\mathbf{v}}$, ${\rm\mathbf{u}}$ and ${\rm\mathbf{d}}$. However, we can observe from $\eqref{Lag4}$ that the augmented objective function is unsolvable due to mutually-coupled exponential terms, which provokes no existing closed-forms. Therefore, we adopt the heuristic scheme for acquiring the sub-optimal solution of QoS-aware RIS adjustment. We iteratively optimize a single RIS element via quantized exhausted search with the remaining elements fixed, and then take turns performing the same procedure until convergence.
\begin{figure}
\vspace{5pt}
\centering
\includegraphics[width=3.3 in]{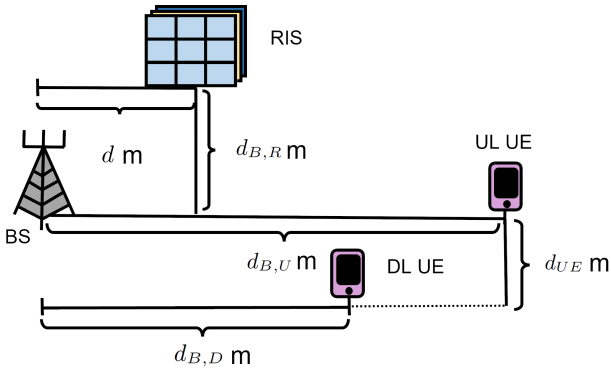}
\caption{The illustration of deployment regarding the relative distance between BS, RIS and DL/UL UE.} \label{Fig.2}
\end{figure}
	Furthermore, for obtaining the closed-form solution, we relax the problem by considering the case without the QoS constraints, i.e., $\lambda_{2K+1}=\lambda_{2K+2}=0$. As explained previously, we can observe from problem $\eqref{obj1}$ that the objective function $\eqref{D1}$ is convex due to its quadratic form. Similarly, we adopt the celebrated Lagrange optimization to acquire the optimal solution. The parameter ${\rm\boldsymbol\lambda}'_{t}$ represents a set regarding the RIS constraint in $\eqref{D1}$. The augmented objective function with ${\rm\boldsymbol\theta}_{t}$ and ${\rm\boldsymbol\lambda}'_{t}$ can be given by
	\vspace{-5pt}
	\begin{equation}
\begin{aligned}\label{Lag}
J({\rm\boldsymbol\theta}_t,{\rm\boldsymbol\lambda}'_t)= P({\rm\boldsymbol\theta}_t,{\rm\boldsymbol\lambda}'_t),
\end{aligned}
\vspace{-5pt}
\end{equation}

which represents that the part of the augmented Lagrangian expression in $\eqref{Lag3}$ without considering QoS constraints. 
Then, based on the theory of linear algebra \cite{17}, we can analytically derive three potential cases on the condition of $N_r+M$ to be smaller, equal to, or greater than the number of RIS elements $K$ in the following theorem, where $N_r+M $ denotes the total receiving antennas in entire system. 
\addtolength{\topmargin}{-0.24 in}
\begin{theorem} \label{thm}
	Consider interference mitigation for RIS-empowered FD MIMO transmission, three potential solutions can be obtained as follows. (a) First, when the number of whole system receiving antennas is equal to RIS elements $N_r+M=K$, the closed-form solution of RIS phase shifts can be derived as\\
	\begin{equation}
\begin{aligned} \label{Uni}
&e^{j\theta_i}= \frac{-\det({\rm\mathbf{W}_c}^{(i)})}{\det(\rm\mathbf{W}_c)}, \quad \forall \, 1\leq i \leq K,
\end{aligned}
\end{equation}
where
\begin{equation}
\left\{
\begin{aligned}
&{\rm\mathbf{W}_c}=\begin{bmatrix} {\rm\mathbf{U}_c}\\{\rm\mathbf{D}_c} \end{bmatrix} \in \mathbb{C}^{N_r\times K},
\\&{\rm\mathbf{W}_c}^{(i)}=\begin{bmatrix} {\rm\mathbf{U}_c}^{(i)}\\{\rm\mathbf{D}_c}^{(i)} \end{bmatrix} \in \mathbb{C}^{M\times K}.
\end{aligned}
\right.
\vspace{-5pt}
\end{equation}

\vspace{20pt}Note that $\det(\cdot)$ is the determinant of a matrix, ${\rm\mathbf{U}_c}= {\rm\mathbf{U}_2}\\\operatorname{diag}({\rm\mathbf{d}_1})\in \mathbb{C}^{N_r\times K}$ and ${\rm\mathbf{D}_c}= {\rm\mathbf{D}_2}\operatorname{diag}({\rm\mathbf{u}_1})\in \mathbb{C}^{M\times K}$.
We define ${\rm\mathbf{U}_c}^{(i)}=[{\rm\mathbf{U}}_{c_1}, ..., {\rm\mathbf{U}}_{c_{i-1}}, {\rm\mathbf{s}}, {\rm\mathbf{U}}_{c_{i+1}}, ...,{\rm\mathbf{U}}_{c_{K}}]$ with ${\rm\mathbf{U}}_{c_i}\in \mathbb{C}^{N_r}$ is the $i$-th column vector of ${\rm\mathbf{U}}_c$, and ${\rm\mathbf{D}_c}^{(i)}=[{\rm\mathbf{D}}_{c_1}, ...,\\ {\rm\mathbf{D}}_{c_{i-1}}, {\rm\mathbf{v}}, {\rm\mathbf{D}}_{c_{i+1}}, ...,{\rm\mathbf{D}}_{c_{K}}]$ with ${\rm\mathbf{D}}_{c_i}\in \mathbb{C}^{M}$ as the $i$-th column vector of ${\rm\mathbf{D}}_c$. (b) When $N_r+M<K$, infinite solutions of RIS phase shifts are obtained as
\begin{equation}\label{uni1}
\left\{
\begin{aligned} 
&{\rm\mathbf{u}_2}_n({\rm\mathbf{d}_1}\circ {\rm\mathbf{p}})+{\rm\mathbf{s}}_n =0,  \quad\forall \, 1\leq n \leq N_r,\\& {\rm\mathbf{d}_2}_n({\rm\mathbf{u}_1}\circ {\rm\mathbf{p}})+{\rm\mathbf{v}}_n =0,  \quad\forall \, 1\leq n \leq M,
\end{aligned}
\right.
\end{equation}

where ${\rm\mathbf{u}_2}_n$ and ${\rm\mathbf{d}_2}_n$ are respectively the $n$-th channel row of ${\rm\mathbf{U}_2}$ and ${\rm\mathbf{D}_2}$. The symbol $\circ$ denotes the Hadamard product, and phase shifts ${\rm\mathbf{p}}= [e^{j\theta_{1}}, e^{j\theta_{2}},..., e^{j\theta_{K}}]^T$. (c)
As for the case of $N_r+M >K$, no optimal solution of phase shift for our considered problem can be acquired.
\vspace{-10pt}
\end{theorem}

\begin{proof} We have $N_r+M $ receiving antennas in our considered system including $N_r$ receiving antennas at BS for UL and $M$ antennas for DL UEs, which means that $N_r+M $ equations are derived based on the first-order derivatives on Lagrangian expression as the total interference received. Accordingly, there exist three conditions as follows: $N_r+M =K$, $N_r+M <K$, and $N_r+M >K$. Based on $\eqref{Lag}$, we can have $K$ first-order derivatives which can be given by
	\begin{equation}
	\begin{aligned}\label{partial}
\frac{\partial J({\rm\boldsymbol\theta}_t,{\rm\boldsymbol\lambda}'_t)}{\partial \theta_i}&= N_i({\rm\boldsymbol\theta}_t,{\rm\boldsymbol\lambda}'_t), \quad \forall \, 1\leq i \leq K.
\end{aligned}
\end{equation} 
\addtolength{\topmargin}{0.35 in}
In the case that $N_r+M =K$, there are $N_r+M $ equations and $K$ variables, and the degree of freedom of variables is equal to the number of equations. Therefore, we can obtain a unique solution of phase shift to null all interference over each channel path. Notice that the optimal solution for phase shifts acquired from $\eqref{Uni}$ does not contain the ratio $\mu$ since all interference can be eliminated. Accordingly, we can obtain the following equivalent equation set from $\eqref{uni1}$ which is given by
\begin{equation}\label{Lag2}
\left\{
\begin{aligned} 
&\sum_{i=1}^{K} {\rm\mathbf{u}}_{2_{n,i}}{\rm\mathbf{d}}_{1_i}e^{{\theta}_i}+{\rm\mathbf{s}}_n =0,  \quad\forall \, 1\leq n \leq N_r,\\& \sum_{i=1}^{K} {\rm\mathbf{d}}_{2_{n,i}}{\rm\mathbf{u}}_{1_i}e^{{\theta}_i}+{\rm\mathbf{v}}_n =0,  \quad\forall \, 1\leq n \leq M.
\end{aligned}
\right.
\end{equation}
Therefore, the phase shift in $\eqref{Uni}$ can be acquired by solving $\eqref{Lag2}$ via substitution among variables of $\theta_i, \forall 1\leq i\leq K$.

Furthermore, consider the second case with $N_r+M<K$, the variable's degree of freedom is higher than the number of equations, which results in infinite solutions. Accordingly, we can derive a decent form of $\eqref{uni1}$ based on $\eqref{Lag2}$. Conversely, if $N_r+M>K$, we have insufficient degrees of freedom to solve these equations which cannot guarantee to reach an optimal solution for phase shift of RIS. This completes the proof. 
\end{proof}

	Based on Theorem 1, we can observe that perfect interference cancellation cannot be performed by RIS if $N_r+M>K$. Note that the considered equation set in $\eqref{partial}$ will still contain the ratio $\mu$. When the value of $\mu$ is large, it implies that we focus on eliminating DL co-channel interference instead of UL SI; conversely, the effect will be opposite. On the other hand, in the case that the number of receiving antenna is equal to or less than that of RIS elements, i.e., $N_r+M \le K$, the adjustment of RIS phase shifts can effectively cancel all the UL SI and DL co-channel interferences. Meanwhile, when $N_r+M=K$, a closed-form solution in $\eqref{Uni}$ can be acquired indicating an optimal RIS phase shift is obtained to perfectly cancel all considered interferences in an RIS-FD MIMO system.
	
\section{Performance Evaluations} \label{PE}	
	The performance results of proposed RFIC scheme for RIS-FD MIMO transmission are evaluated via simulations. The channels comprise distance-aware large-scale and small-scale Rayleigh fading. The large-scale pathloss of $3.5$ GHz frequency spectrum is defined as $PL_{\text{LoS}}=38.88+22\log_{10}(d_0)$ dB for line-of-sight path \cite{18}, where $d_0$ is the distance between the transmitter and receiver. The Rayleigh fading is considered as an exponential distribution with expectation of one. The main system parameters of our RIS-FD MIMO system are listed in Table \ref{Parameter}. The transmit antennas of UL UEs and BS are assigned with equal power and the total power is respectively denoted by $P_{U,max}=N\cdot{p_U}$ and $P_{D,max}=N_t\cdot{p_D}$. The simulated scenario is shown in Fig. \ref{Fig.2}, where $d$, $d_{B,U}$ and $d_{B,D}$ denote the horizontal distance of BS to RIS, BS to UL UEs and BS to DL UEs, respectively; while $d_{B,R}$ and $d_{UE}$ respectively indicates the distance between BS and UL UEs and between DL UEs to UL UEs. We set $d=d_{B,U}=d_{B,D}=60$ m, $d_{B,R}=70$ m, and $d_{UE}=60$ m. Note that we evaluate received DL/UL sum rate considering different transmit power, horizontal distance between RIS and BS,   and the numbers of UL UEs, transmit antennas and RIS elements. The proposed RFIC scheme is also compared to existing benchmarks in open literature.


\begin{table}
\begin{footnotesize}
\begin{center}
\caption {Parameters of RIS-FD MIMO System}
    \begin{tabular}{ll}
        \hline
        System Parameters & Value \\ \hline 
        Carrier frequency & $3.5$ GHz \\
        Noise power spectral density & $-174$ dBm/Hz\\
        Bandwidth & $20$ MHz \\
		Total transmit power of UL UEs & $\left[ 1, 15\right]$ mWatt\\
		Total transmit power of BS & $\left[ 1, 15\right]$ mWatt\\
		RIS reflection decaying efficiency & $0.95$ \\
		Number of transmit/receiving antennas & $\{2,4\}$\\
		Number of RIS elements & $\{0, 4, 8\}$\\
		Number of UL UEs & $\left[1,5\right]$ \\
		Number of DL UEs & $2$\\
        \hline
    \end{tabular} \label{Parameter}
\end{center}
\end{footnotesize}
\vspace{-10pt}
\end{table}

\begin{figure}[h]

\centering
\includegraphics[width=3.3in]{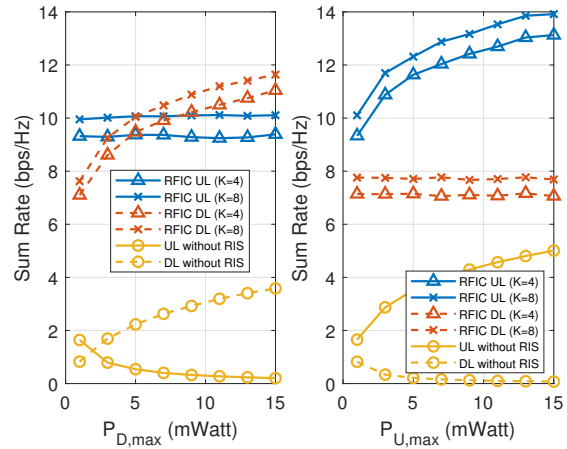}
\caption{Sum rate of our proposed RFIC versus different total transmit power $P_{D,max}$ and $P_{U,max}$ for $N_r=2$ receiving antennas, $N_t=2$ transmit antennas, $N=2$ UL UEs and $M=2$ DL UEs.} \label{Fig.3}
\end{figure}	

	In Fig. \ref{Fig.3}, we evaluate sum rate of proposed RFIC for non-QoS case with $t_{thr,D}=t_{thr,U}=0$ versus different transmit power $P_{D,max}$ of BS and $P_{U,max}$ of UL UE for $N_r=2$, $N_t=2$ and $N=2$. In the left plot, as $P_{D,max}$ increases, it can be observed from the cases without RIS that the DL sum rate is intuitively increased; while that for UL is decreased since it becomes SI power for UL. On the other hand, the UL sum rate remains unchanged since our proposed RFIC can null the SI and DL sum rate due to larger signal power $P_{D,max}$. In the right plot, with similar reason, the RFIC scheme results in enhanced and unchanged sum rate for UL and DL sum rate respectively with an increase in total transmit power of UL UE $P_{U,max}$. With the increment of RIS elements from $K=4$ to $8$, higher sum rate can be achieved due to more reflected desired signals. Compared with the cases without RIS, it can be seen that those with RFIC achieve much higher sum rates with its capability of interference cancellation via the adjustment of RIS phase shifts.

\begin{figure}
\centering
\includegraphics[width=3.3in]{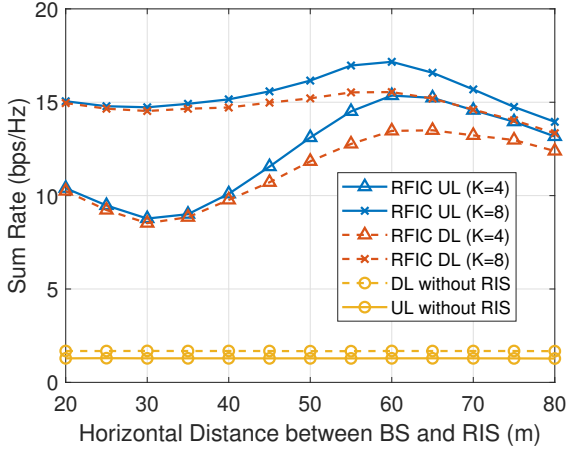}
\caption{Sum rate of proposed RFIC versus different distances between BS and RIS with $N_t=2$ transmit antennas, $N_r=2$ receiving antennas, $K=\{4,8\}$ RIS elements, $N=2$ UL UEs, $M=2$ DL UEs and $P_{D,max} = P_{U,max}=1$ mWatt.} \label{Fig.4}

\end{figure}

	Fig. \ref{Fig.4} depicts the sum rate performance of our proposed RFIC considering $t_{thr,D}=t_{thr,U}=0$ versus different horizontal distances from BS to RIS under $N_r=2$, $N_t=2$ and $N=2$. We can observe from Fig. \ref{Fig.4} that there exist three performance trends: $d<60$ m, $d=60$ m and $d>60$ m. If $d$ is smaller than $60$ m, the benefit of RIS deployment is significantly high, which can reflect more desired signal power with cancellation of severe interference. Moreover, it performs a convex curve when $d<60$ m since the reflected desired signals are the weakest when $d=30$ m, which is caused by the longest transmit distance through RIS. If $d$ is equal to $60$ m, the highest performance can be reached due to the shortest distance between RIS and UEs. As $d$ becomes larger than $60$ m, sum rates are decreasing due to significant effects of higher pathloss. 

\begin{figure}

\centering
\includegraphics[width=3.3in]{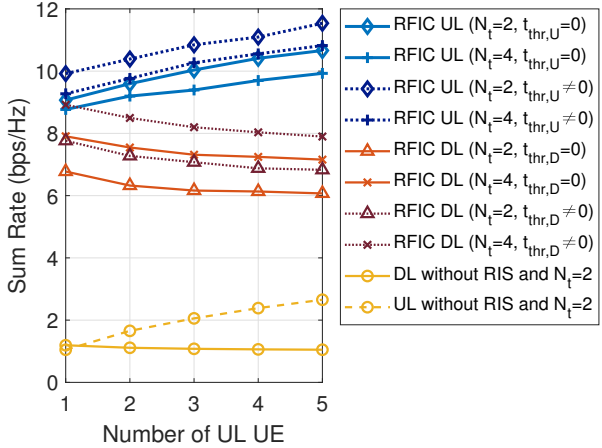}
\caption{Sum rate versus different numbers of UL UEs with $N_t=\{2,4\}$ transmit antennas, $N_r=2$ receiving antennas, RIS elements of $K=\{4,8\}$, $M=2$ DL UEs and $P_{D,max} = P_{U,max}=1$ mWatt.} \label{Fig.5}
\end{figure}

	In Fig. \ref{Fig.5}, we demonstrate sum rate of  RFIC scheme with $t_{thr,D}=0$ and $t_{thr,U}=0$ for non-QoS case and $t_{thr,D}=70$ and $t_{thr,U}=70$ for QoS-aware case versus different numbers of UL UEs. It reveals that UL sum rate increases as the number of UL UEs increases due to more signal power reflected, whilst the DL sum rate decreases because it causes more co-channel interference for DL. Furthermore, it can be observed that additionally deployed transmit antennas $N_t$ will cause performance degradation for UL but result in increased DL sum rate since transmit signals from the BS will induce not only SI but also DL transmit signal. Moreover, the results show that proposed RFIC scheme with QoS-aware cases outperforms those with non-QoS due to joint consideration of guaranteed received signal quality and mitigation of interference. Although RFIC cannot perfectly cancel total interference under the case with QoS restriction, it alleviates most of interference in order to sustain desired signal quality.

\begin{figure}

\centering
\includegraphics[width=3.3in]{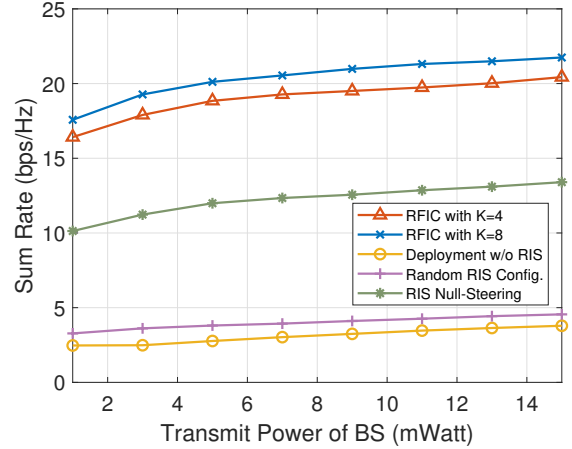}
\caption{The comparison of total sum rate performance of proposed RFIC scheme with existing benchmarks of random RIS configuration, null-steering method, deployment without RIS by considering different transmit power of BS, $N_t=2$ transmit antennas, $N_r=2$ receiving antennas, $M=2$ DL UEs, $N=2$ UL UEs and $P_{U,max}=1$ mWatt.} \label{Fig.6}
\vspace{-10pt}
\end{figure}

	In Fig. \ref{Fig.6}, we compare the total sum rate performance  of our proposed RFIC with existing benchmarks versus different transmit power $P_{D,max}$ of BS with $N_t=2$, $N_r=2$ and $N=2$. We consider three benchmarks including conventional deployment without RIS, random RIS configuration, and null-steering method \cite{100}. Note that random RIS method means arbitrary configuration of phase shifts under the constraints of $[0,2\pi]$. Moreover, null-steering is regarded as a well-known postcoder method which minimizes interference after received superposed signals. We can observe from Fig. \ref{Fig.6} that the performance of random RIS setting is similar to that without RIS deployment since it does not provide adequate capability for interference mitigation. Furthermore, null-steering has higher performance than random phase shifts and non-RIS cases. This is because null-steering is capable of mitigating the interference but still with abundant residual interference affecting the desired signal. As a result, the proposed RFIC can perfectly mitigate interference, outperforming the other exiting schemes in open literature, e.g., more than $50\%$ of performance gain of sum rate is achieved by RFIC with $K=8$ having $22$ bps/Hz of sum rate compared to RIS null-steering with $14$ bps/Hz under $P_{D,max} =15$ mWatt.

\section{Conclusions} \label{CON}
	In this paper, we conceive an RFIC scheme for RIS-empowered interference mitigation for 6G FD MIMO transmissions considering QoS requirement of desired signal. The proposed RFIC scheme is theoretically proved to posses a closed-form solution of RIS phase shift adjustment when the number of RIS elements and receiving antennas are equivalent. There exist infinite solutions if the number of RIS elements is higher than receiving antennas; conversely, no optimal outcome can be obtained. Simulation results have shown that the proposed RFIC scheme effectively mitigates the SI of FD MIMO system with more equipped RIS elements compared to that without RIS deployment. Moreover, the highest sum rate can be attained when the UEs are located nearest to RIS. Benefited by substantial interference cancellation, the proposed RFIC outperforms the other existing methods including deployment without RIS, random configuration and null-steering method in terms of the highest performance.
	
\renewcommand{\IEEEbibitemsep}{0pt plus 0.5pt}
\makeatletter
\IEEEtriggercmd{\reset@font\normalfont\fontsize{7pt}{7pt}\selectfont}
\makeatother
\IEEEtriggeratref{1}

\bibliographystyle{IEEEtran}
\bibliography{IEEEabrv}

\end{document}